\documentclass[journal]{IEEEtran}
\usepackage[latin1]{inputenc}
\usepackage{graphicx}
\usepackage{amssymb,amsmath,amsfonts,amsthm}
\usepackage{algorithm,algorithmic}
\usepackage{cite}
\usepackage{float}
\usepackage{graphics}
\usepackage{subfigure}
\usepackage{xspace}
\usepackage[usenames,dvipsnames]{color}
\usepackage{epsfig, hyperref}

\newtheorem{theorem}{Theorem}

\begin{document}

\title{Wireless Power Transfer Under Kullback-Leibler Distribution Uncertainty: A Mathematical Framework}

\author{Ioannis Krikidis,~\IEEEmembership{Fellow,~IEEE}  
\thanks{This work has received funding from the European Research Council (ERC) under the European Union's Horizon 2020 research and innovation programme (Grant agreement No. 819819).}
\thanks{I. Krikidis is with the Department of Electrical and Computer Engineering, Faculty of Engineering, University of Cyprus, Nicosia 1678 (E-mail: {\sf krikidis@ucy.ac.cy}).}}

\maketitle
\vspace{-0.8cm}
\begin{abstract}
In this letter, we study the performance of a wireless power transfer system under  energy harvesting distribution uncertainty. The uncertainty captures practical nonidealities of the rectification process and is modelled as the maximum Kullback-Leibler distance of the actual distribution from a nominal distribution. The case where symmetrized divergence is considered for the statistical distance between distributions is also considered. By formulating a convex optimization problem, we investigate a mathematical framework that provides closed form expressions for the minimum average harvested energy (the worst performance) and the associated statistical distribution. Theoretical results show that the energy harvesting performance is significantly degraded as the level of the uncertainty increases.
\end{abstract}
\vspace{-0.6cm}
\begin{keywords}
Wireless power transfer, uncertainty, Kullback-Leibler divergence, average energy harvested.
\end{keywords}

\vspace{-0.1cm}
\section{Introduction}

\IEEEPARstart{W}{ireless} power transfer (WPT) is a new technology which allows low-power devices to harvest energy from dedicated/ambient radio-frequency signals. From the seminal work in \cite{RUI2}, where the authors introduced the principles of WPT and the basic network architectures, WPT has attracted a lot of attention from both academia and industry. Recent studies take into account the nonlinearity of the rectification circuit and study WPT from information theory, signal-processing and/or networking perspective \cite{CLE}.  

One of the most fundamental questions in WPT is the modelling of the energy harvesting (EH) process. To address this question, several deterministic models have been proposed in the literature by trading off simplicity and accuracy. From the initial linear model which has been used in \cite{RUI2}, the piece-wise linear model captures (in a simple way) the three operation regions of the rectification circuit \cite{LOP}. On the other hand, more sophisticated parametric nonlinear functions (e.g., sigmoidal \cite{BOS}, fractional \cite{CHE}, etc.) have been proposed, where the parameters are tuned by using real data and curve fitting tools. Although these harvesting models try to approximate the behaviour of the circuit from a communication theory perspective, they neglect important practical phenomena e.g., antenna mismatching, parasitic effects, RC filter etc \cite{PAN}. Therefore, these deterministic models are associated with specific operation points of the rectification circuit and are not able to capture its general behaviour.  To make these models more accurate and enhance their practical interest, they should be extended by capturing uncertainty; for systems with low computation/processing capabilities (e.g.,  WPT), this uncertainty can be modelled by a partial knowledge of the EH statistical distribution (including wireless channel and EH process) \cite{LAP}.

A fundamental model for communication under channel distribution uncertainty is the compound channel \cite[Sec. III]{LAP}, where the transmitter knows that the (unknown) channel distribution is within a given Kullback-Leibler (KL) divergence (uncertainty) from a nominal distribution. In this case, the Shannon capacity is associated to a class of fading distributions  rather than a specific distribution. The work in \cite{CHA} studies the compound outage probability of a communication system under KL fading distribution uncertainty and identifies two fundamental operation regimes. The same mathematical framework is used in \cite{CHA2} to study the ergodic compound capacity of a multiple-input multiple-output system with incomplete channel state information and in \cite{CHA3} to design the optimal control of a stochastic system. The work in \cite{BIG}  extends this framework for more general objective functions and investigates the robustness of modulation and coding schemes in body-area networks.

The objective of our work is to propose a mathematical framework that integrates an EH distribution uncertainty in WPT systems. Inspired by the fundamental concept of compound channel \cite{CHA,BIG}, we study the minimum average harvested energy for a basic point-to-point link when the actual EH distribution is within a certain KL uncertainty from a nominal distribution. By formulating a convex optimization problem with respect to the actual distribution, we admit a general closed-form  solution via the evaluation of the Karush-Kuhn-Tucker (KKT) optimality conditions. We study the regular KL divergence metric and its asymmetric counterpart  ($\mathcal{D}_{\text{KL}}(f_0||f)$ and $\mathcal{D}_{\text{KL}}(f||f_0)$) as well as the symmetrized divergence. We demonstrate that the EH performance (average harvested energy) significantly decreases as the uncertainty distance increases. Simplified closed-form expressions are derived for the linear EH model with Rayleigh fading, where the nominal distribution is exponential. The case where the class/type of the actual EH distribution is known at the transmitter is also discussed and simple expressions are proposed. The proposed mathematical framework is general and can be used as a basis to systematically study WPT applications.   
  
\vspace{-0.29cm}  
\section{Average harvested energy under uncertainty}\label{sec2}

We study a fundamental point-to-point link consisting of one transmitter that transmits energy signals to a single receiver.  Due to the wireless channel and the practical nonidealities of the rectenna circuit \cite{PAN}, the energy harvested becomes a random variable with a nominal distribution $f_0(x)\geq 0$ and an unknown true distribution $f(x)\geq 0$. According to the principles of compound channel \cite{LAP}, we assume that this nominal distribution is within a KL divergence $\mathcal{D}_{\text{KL}}(f_0||f)\leq d$ from the actual distribution, where $\mathcal{D}_{\text{KL}}(f||g)=\int f(x)\log \frac{f(x)}{g(x)}dx$ denotes the KL distance between the distributions $f(x)$ and $g(x)$, and $d$ is the maximum KL divergence; both $f_0(x)$ and $d$ are known at the transmitter through an appropriate training/estimation process and feedback channel\footnote{The nominal distribution $f_0(x)$ refers to the selected model, while the parameter $d$ represents prior knowledge of the approximate behaviour of the distribution and highly depends on the size/quality of the available data, the estimation mechanism \cite{MCK} (e.g., dynamics perturbations, hypothesis test, etc.) as well as the technical capabilities of  the communication and WPT infrastructure (i.e., computational resources, feedback channel, rectenna circuit, etc.) \cite{CHA,CHA2}.}. It is worth noting that the parameters $f_0(x)$ and $d$ determine the EH uncertainty; they are considered constant for the whole communication and independent on the transmit parameters. If the performance of the WPT system is characterized by the average harvested energy, the worst case scenario is captured by the following optimization problem 
\vspace{-0.1cm}
\begin{align}
&(P1)\;\;\min_f \int x f(x)dx \\
&\;\;\;\;\;\;\;\;\;\;\text{s. t.}\; \mathcal{D}_{\text{KL}}(f_0||f)\leq d,\;\;\;\int f(x)dx=1. \label{co1}
\end{align}
The solution of the optimization problem (P1) gives the minimum energy that can be harvested, when there is KL uncertainty for the EH distribution. This is an important information for the design of a WPT system with critical quality of service constraints (e.g., design a robust transmission strategy that ensures a minimum average harvested energy for all possible cases; for the whole set of uncertainty). If $\mathcal{E}=\int x f(x)dx$ denotes the solution to (P1), the following theorem gives the associated actual distribution.
\vspace{-0.2cm}
\begin{theorem}\label{th1}
The actual distribution that achieves the minimum average harvested energy is $f(x)=\frac{1}{q(\mu^*)}\frac{f_0(x)}{x+\mu^*}$, where $\mu^*$ is the solution of the ($1$- dimensional) equation $\int f_0(x) \log[q(\mu)(x+\mu)]dx=d$ and $q(\mu)=\int \frac{f_0(x)}{x+\mu}dx$.
\end{theorem}
\begin{proof}
The proof is given in Appendix \ref{ap1}.
\end{proof}
We now consider the performance of the system for two asymptotic cases i.e., $d=0$ and $d\rightarrow \infty$. For the case where $d=0$, there is not distribution uncertainty and thus $f(x)=f_0(x)$ and $\mathcal{E}_0=\int x f_0(x)dx$ (expectation of the nominal distribution). On the other hand, for the case $d\rightarrow \infty$, we have $\int f_0(x) \log[q(\mu)(x+\mu^*)]dx\rightarrow \infty$ which after simple manipulations\footnote{We have $\int f_0(x)\log q(\mu^*)+\int f_0(x)\log(x+\mu^*)\leq \log q(\mu^*) +\mathcal{E}_0+\mu^*$, where we used the basic inequality $\log(x)\leq x$; since $\mu^*$ and $\mathcal{E}_0$ are finite, we have $q(\mu^*)\rightarrow \infty$.} gives $\mu^*\rightarrow 0$, $q(\mu^*)\rightarrow \infty$ and thus $\mathcal{E}\rightarrow 0$. In the following discussion, we apply Theorem \ref{th1} for the case where the nominal distribution is exponential\footnote{The exponential distribution is used for the sake of simplicity to introduce the proposed mathematical framework. It also refers to a linear/piece-wise linear EH model with Rayleigh channel fading.}. 

\noindent {\it Special Case (Exponential distribution):} For the case where the nominal distribution is exponential i.e., $f_0(x)=\lambda_0 \exp(-\lambda_0 x)$, we have
\begin{align}
&f(x)=\frac{1}{q(\mu^*)}\frac{\lambda_0\exp(-\lambda_0 x)}{x+\mu^*}=\frac{\exp(-\lambda_0(x+\mu^*))}{E_1(\lambda_0 \mu^*)(x+\mu^*)}, \\
&q(\mu^*)=\int_0^{\infty}\frac{\lambda_0\exp(-\lambda_0 x)}{x+\mu^*}dx=\lambda_0 \exp(\lambda_0 \mu^*)E_1(\lambda_0 \mu^*), \label{q1} \\
\mathcal{E}&=\int_{0}^{\infty}xf(x)dx=\frac{1}{\lambda_0  \exp(\lambda_0 \mu^*)E_1(\lambda_0 \mu^*)}-\mu^*, \label{q2}
\end{align}
where \eqref{q1}, \eqref{q2} employ the expressions in \cite[3.352.4]{GRA}, \cite[3.353.5]{GRA}, respectively, and $E_1(x)=\int_{x}^{\infty}\frac{\exp(-t)}{t}dt$ is the exponential integral \cite{GRA}. The parameter $\mu^*$ is the solution of the simplified equation $q(\mu)/\lambda_0 +\log (\mu q(\mu))=d$. We can also calculate the cumulative density function (CDF) of the distribution, which is useful for the evaluation of the energy outage probability. Specifically, the CDF of the actual distribution is equal to \cite[3.352.1]{GRA}
\begin{align}
F(x)=\!\mathbb{P}(X\leq x)=\!\!\int_{0}^{x}\!\!f(x)dx=\!1\!-\!\frac{E_1(\lambda_0(x+\mu^*))}{E_1(\lambda_0 \mu^*)}.
\end{align}

\vspace{-0.5cm}  
\subsection{KL divergence asymmetry}

One of the main properties of the KL divergence is that is not symmetric i.e., $\mathcal{D}_{\text{KL}}(f||g) \neq \mathcal{D}_{\text{KL}}(g||f)$. Here, we replace the KL constraint in \eqref{co1} with $\mathcal{D}_{\text{KL}}(f||f_0)\leq d$ i.e., 
\begin{align}
&(P2)\;\;\min_f \int x f(x)dx \\
&\;\;\;\;\;\;\;\;\;\;\text{s. t.}\; \mathcal{D}_{\text{KL}}(f||f_0)\leq d,\;\;\;\int f(x)dx=1. \label{co2}
\end{align}
By using a similar mathematical framework with the formulation in (P1), we state the following theorem.
\begin{theorem}\label{th2}
The actual distribution that achieves the minimum average harvested energy is $f(x)=\frac{\exp(-x/s^*)f_0(x)}{\int \exp(-x/s^*)f_0(x)dx}$, where the variable $s^*$ is given numerically by solving the ($1$-dimensional) equation in \eqref{s1}.
\end{theorem}
\begin{proof}
The proof is given in Appendix \ref{ap2}.
\end{proof}
By using similar arguments with the problem (P1), we can see that $f(x)=f_0(x)$ when $d=0$; the case $d\rightarrow \infty$ will be discussed bellow for a specific example of nominal distribution. 

\noindent {\it Special Case (Exponential distribution):} For the case where $f_0(x)=\lambda_0 \exp(-\lambda_0 x)$, we have $f(x)=(1/s^*+\lambda_0) \exp(-x(1/s^*+\lambda_0))$, which shows that the actual distribution is also exponential with parameter $(1/s^*+\lambda_0)$. In addition, we have $\zeta(s^*)=\frac{\lambda_0}{(\lambda_0+1/s^*)^2}$, $\psi_1(s^*)=\frac{\lambda_0}{(\lambda_0+1/s^*)}$ and thus $\mathcal{E}=\frac{\zeta(s^*)}{\psi_1(s^*)}=\frac{s^*}{1+s^* \lambda_0}$; $s^*$ is the solution of the simplified nonlinear equation $\xi(s)=d$ where $\xi(s)=s\log \frac{Z(s)}{\lambda_0}-\frac{1}{Z(s)}$ and $Z(s)=\frac{1}{s}+\lambda_0$ with $s \in (0, \infty)$.

To study the asymptotic performance, we can see that the function $\xi(s)$ is increasing ($\xi'(s)\geq 0$) in the interval $(0, \overline{s}^*)$ and decreasing  ($\xi'(s)\leq 0$) in the interval $(\overline{s}^*, \infty)$, where $\xi'(s)=\log\left(1+\frac{1}{\lambda_0 s}\right)-\frac{2+\lambda_0 s}{(\lambda_0 s+1)^2}$, and $\xi'(\overline{s}^*)=0$. Therefore, for $d\geq \xi'(s)$ we have $\mathcal{E}=\frac{\overline{s}^*}{1+\overline{s}^*\lambda_0}$; for $d\rightarrow 0$, we have $\lim_{s\rightarrow \infty}\xi(s)\rightarrow 0$ and thus $s^*\rightarrow \infty$ which gives $\mathcal{E}=1/\lambda_0$. This is the basic difference between the two KL divergence metrics; $\mathcal{D}_{\text{KL}}(f||f_0)$ converges to a nonzero constant floor, while $\mathcal{D}_{\text{KL}}(f_0||f)$ asymptotically converges to zero. 

\vspace{-0.3cm}
\subsection{Symmetrized divergence}\label{sd}

The symmetrized divergence is another statistical distance metric, which in contrast to the KL divergence metrics, it satisfies the property of symmetry \cite{MOU}. Specifically, the symmetrized divergence is defined as 
\begin{align}
\mathcal{D}_{\text{SYM}}(f_0, f)=\frac{1}{2}\big[\mathcal{D}_{\text{KL}}(f_0||f)+\mathcal{D}_{\text{KL}}(f||f_0) \big],
\end{align}
and it can be seen that $\mathcal{D}_{\text{SYM}}(f_0, f)=\mathcal{D}_{\text{SYM}}(f, f_0)$; we note that $\mathcal{D}_{\text{SYM}}$ might not satisfy the triangle inequality and therefore it is not an actual distance. In this case, the original optimization problem is formulated as follows 
\begin{align}
&(P3)\;\;\min_f \int x f(x)dx \\
&\;\;\;\;\;\;\;\;\;\;\text{s. t.}\; \mathcal{D}_{\text{SYM}}(f_0, f)\leq d,\;\;\;\int f(x)dx=1. \label{co3}
\end{align}
The optimization problem (P3) is also convex since the symmetrized divergence is a convex statistical metric. In the following theorem, we provide the actual distribution that solves (P3) as a function of the nominal distribution.  
\begin{theorem}\label{th3}
The actual distribution that achieves the minimum average harvested energy is $f(x)=\frac{f_0(x)}{W_0 \left(\exp \left(\frac{2(x+\mu^*)}{s^*} \right) \right)}$, where the parameters $s^*>0$, $\mu^*$ are computed numerically by solving the (two-dimensional) system of equations in \eqref{x1}, \eqref{x2}.   
\end{theorem}
\begin{proof}
The proof is given in Appendix \ref{ap3}.
\end{proof}

\vspace{-0.3cm}
\subsection{Known class/type of the true distribution- case studies}

Here, we assume that although there is uncertainty for the distribution $f(x)$, we know the exact class/type of the unknown distribution. This assumption can be supported by a more strict distribution estimation process, which provides the type of the distribution in addition to the maximum statistical distance $d$; this case has mainly theoretical interest and allows to further investigate the impact of the divergence between the selected and the true model on the system performance. The following discussion refers to two specific case studies. 
 
\noindent {\bf Case study 1} (Exponential distributions): if the nominal and the actual distribution are exponential with parameters $\lambda_0$ and $\lambda_1\geq \lambda_0$, respectively, the KL divergence is written as
\begin{align}
\mathcal{D}_{\text{KL}}(f_0||f)&=\!\int_{0}^{\infty}\!\lambda_0 \exp(-\lambda_0 x)\log \frac{\lambda_0 \exp(-\lambda_0 x)}{\lambda_1 \exp(-\lambda_1 x)}dx \nonumber \\ 
&= \log \frac{\lambda_0}{\lambda_1}+\frac{\lambda_1-\lambda_0}{\lambda_0}.
\end{align}
In this case, we have $\mathcal{D}_{\text{KL}}(f_0||f)\leq d\Rightarrow \frac{\lambda_1}{\lambda_0}\exp(-\frac{\lambda_1}{\lambda_0})\geq \exp(-d-1)$. Since the function $g(x)=x\exp(-x)$ is monotonically decreasing for $x\geq 1$ (i.e., $g'(x)\leq 0$ with $x\geq1$), the minimum average expected value is achieved at the boundary i.e., $\frac{\lambda_1}{\lambda_0}\exp(-\frac{\lambda_1}{\lambda_0})=\exp(-d-1)$ which gives $\lambda_1^*=\max[-\lambda_0W_0(-\exp(-d-1)),\lambda_0(1+d)]$, where $W_0(\cdot)$ is the principal branch of the LambertW function \cite{GRA}; in a similar way, for the  symmetric case ($\mathcal{D}_{\text{KL}}(f||f_0)$) we have $\lambda_1^*=\max[-\lambda_0/W_0(-\exp(-1-d)),\lambda_0/(1+d)]$.
For the case of the symmetrized divergence, we have $\mathcal{D}_\text{SYM}(f_0,f)=\frac{\lambda_1}{\lambda_0}+\frac{\lambda_0}{\lambda_1}-2$ and therefore $\mathcal{D}_\text{SYM}(f_0,f)\leq d \Rightarrow \lambda_1^2-2\lambda_0 (d+1)\lambda_1+\lambda_0^2\leq 0$ with solution $\lambda_1 \in [\lambda_0(d+1)-\lambda_0\sqrt{d(d+2)},\lambda_0(d+1)+\lambda_0\sqrt{d(d+2)}]$; therefore $\lambda_1^*=\lambda_0(d+1)+\lambda_0\sqrt{d(d+2)}$. For all the cases, the minimum  average harvested energy is equal to $\mathcal{E}=1/\lambda_1^*$. 

\noindent {\bf Case study 2} (Uniform distributions): If the nominal and the actual distributions are uniform over the interval $[0,\alpha]$ and $[0,\beta]$ with $\beta\leq \alpha$, respectively, we obtain 
$\mathcal{D}_{\text{KL}}(f||f_0) \leq d\Rightarrow \log\frac{\alpha}{\beta} \leq d \Rightarrow \beta\geq \alpha \exp(-d)$,
and therefore the solution of the optimization problem considered is achieved at the boundary i.e., $\mathcal{E}=\frac{\alpha\exp(-d)}{2}$. 
As for KL asymmetry and the symmetric divergence, $f_0(x)$ is not dominated by $f(x)$ and thus we have $\mathcal{D}_{\text{KL}}(f_0||f)=\infty$ \cite[Def. 6.1]{MOU}; therefore, these two metrics have not practical interest.

\begin{figure}[t]
\centering
\includegraphics[width=0.8\linewidth]{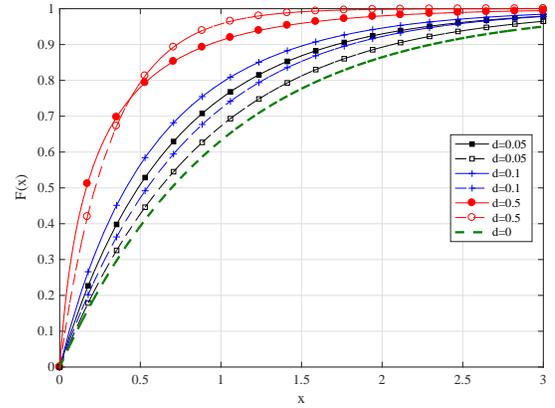}
\caption{CDF of the actual distribution for different values of $d$; solid-lines and dashed-lines correspond to the KL divergence $\mathcal{D}_{\text{KL}}(f_0||f)$ and $\mathcal{D}_{\text{KL}}(f||f_0)$, respectively; nominal distribution $f_0(x)=\exp(-x)$.}\label{fig0}
\end{figure}      

\begin{figure}[t]
\centering
\includegraphics[width=0.75\linewidth]{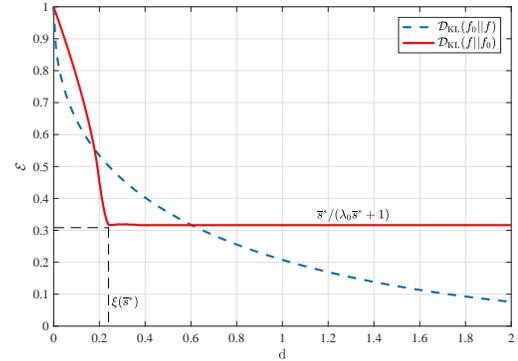}
\caption{Minimum average energy harvested (worst case performance) versus $d$ with $f_0(x)=\exp(-x)$; for the $\mathcal{D}_{\text{KL}}(f||f_0)$ case,  we have $\xi'(s)=0\Rightarrow \overline{s}^*=0.46$, and $\mathcal{E}=\frac{\overline{s}^*}{1+\lambda_0 \overline{s}^*}=0.31$.}\label{fig1}
\end{figure}

\begin{figure}[t]
\centering
\includegraphics[width=0.85\linewidth]{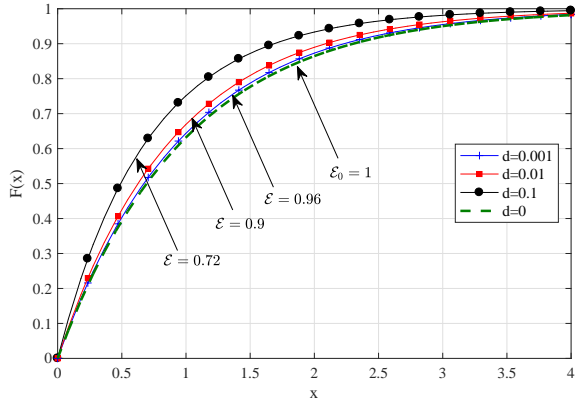}
\caption{CDF of the actual distribution under symmetrized divergence uncertainty for different values of $d$; nominal distribution $f_0(x)=\exp(-x)$.}\label{fig2}
\end{figure}

\begin{figure}[t]
\centering
\includegraphics[width=0.65\linewidth]{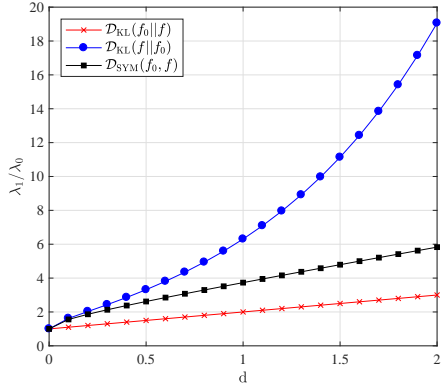}
\caption{Uncertainty with known class of distribution- $\lambda_1/\lambda_0$ versus the distance $d$ for an exponential nominal distribution $f_0(x)=\lambda_0 \exp(-\lambda_0 x)$.}\label{fig3}
\end{figure} 

\vspace{-0.3cm}
\section{Numerical results}     
Computer simulations are carried-out to validate the proposed mathematical framework. For the sake of simplicity, the numerical results concern a normalized exponential nominal distribution i.e., $f_0(x)=\lambda_0 \exp(-\lambda_0 x)$ with $\lambda_0=1$.   
      
Fig. \ref{fig0} shows the CDF of the actual distribution for different values of $d$ for both KL divergence metrics considered. As it can be seen, a higher KL distance $d$ shifts the CDF curves from right to left and increases the gap between nominal and actual distribution. We can see that the uncertainty significantly affects the harvested energy and an increase of the distance $d$ results in a more conservative EH performance. Another interesting observation is that both KL divergence metrics follow the same behaviour/trend and result in same general conclusions; however, it worth noting that $\mathcal{D}_{\text{KL}}(f||f_0)$ results in a more conservative actual distribution than $\mathcal{D}_{\text{KL}}(f_0||f)$ for low $d$. As the distance $d$ increases, we have the opposite behaviour which is also justified by the asymptotic performance of the two metrics (see Fig. \ref{fig1} for $d=0.5$).  

Fig. \ref{fig1} plots the minimum average harvested energy (worst case performance) versus the KL distance $d$ for both considered KL divergence metrics. As it can be seen, the worst performance significantly decreases as the distance $d$ increases. This plot also validates our asymptotic results for high values of $d$ (i.e., $d\rightarrow \infty$) which are discussed in Section \ref{sec2}. It can be seen that swapping the distributions (nominal and true) results in a floor effect; $\mathcal{E}$ asymptotically converges to zero for the metric $\mathcal{D}_{\text{KL}}(f_0||f)$, while $\mathcal{E}$ converges to a constant floor that depends on the KL distance $d$ i.e., $\overline{s}^*/(\lambda \overline{s}^*+1)$ for the metric $\mathcal{D}_{\text{KL}}(f||f_0)$.

In Fig. \ref{fig2}, we show the CDF of the actual distribution for different values of $d$, when statistical distance corresponds to the symmetrized divergence metric. It can be seen that the CDF is shifted towards its left as the distance $d$ increases; therefore the associated energy outage probability increases. This behaviour is inline with the observations in Fig. \ref{fig0}. However, it is worth noting that the symmetrized divergence results in a CDF which is between the CDFs associated with  $\mathcal{D}_{\text{KL}}(f_0||f)$ and $\mathcal{D}_{\text{KL}}(f||f_0)$. Therefore, the symmetrized divergence consists of a balance between the two KL metrics considered. In this figure, we also show the minimum harvested energy (solution to (P3)) by following the theoretical framework in Section \ref{sd}.

Finally, Fig. \ref{fig3} deals with the actual distribution when the class/type of the true distribution is known. Specifically, we assume that nominal/actual distributions are exponential with parameters $\lambda_0$ and $\lambda_1$, respectively, and we plot $\lambda_1/\lambda_0$ versus $d$ for the statistical divergence metrics considered. It can be seen that the difference between the true and the nominal distributions increases as the divergence $d$ increases. In addition, we can see $\mathcal{D}_{\text{KL}}(f_0||f)$ provides a better estimation of $f(x)$ (i.e., $\lambda_1/\lambda_0$ increases linearly with $d$), while  $\mathcal{D}_{\text{KL}}(f||f_0)$ results in an exponential difference. 

\vspace{-0.2cm}
\section{Conclusion}

A mathematical framework that integrates EH uncertainty in the current deterministic WPT models has been proposed. By exploiting the notion of the compound channel model, we have studied the worst performance (average harvested energy) when the actual end-to-end distribution is within a given KL/symmetrized maximum statistical distance from a nominal distribution. General closed-form expressions that hold for any nominal distribution as well as simplified expressions that refer to an exponential nominal distribution have been derived. Theoretical results show that distribution uncertainty significantly affects EH performance and therefore its integration to current WPT models is essential for a robust design.         

\vspace{-0.1cm}

\appendices 
\section{Proof of Theorem \ref{th1}} \label{ap1}
\vspace{-0.1cm}

Since the problem is convex (i.e., linear objective function and convex constraints), KKT conditions are necessary and sufficient for optimality \cite{BOY}. The Lagrangian function of the problem is written as 
\begin{align}
&L=\int x f(x) dx \nonumber \\
&\!+\! s \left(\int\!\!\ f_0(x)\log \frac{f_0(x)}{f(x)}dx-d\right)+\mu \left(\int\!\! f(x)dx-1\right),
\end{align}
where $s$ and $\mu$ are the Lagrange multipliers of two constraints.  The associated KKT conditions are given as follows \cite{BOY}
\begin{align}
 x-s\frac{f_0(x)}{f(x)}+\mu=0&, \label{eq1}  \\
\int f(x)dx-1 =0&, \label{eq2} \\
s \left(\int f_0(x)\log \frac{f_0(x)}{f(x)}dx-d \right)=0&, \label{eq3} \\
s\geq 0&.
\end{align}
\noindent From the complementary slackness condition in \eqref{eq3}, for $s>0$ the minimum is achieved at the boundary. In this case, by combining \eqref{eq1} and \eqref{eq2} and after some manipulations, we have
\begin{align}
f(x)=\frac{f_0(x)}{q(\mu^*)(x+\mu^*)},
\end{align}
\vspace{-0.2cm}
where
\vspace{-0.2cm}
\begin{align}
q(\mu)=\int \frac{f_0(x)}{x+\mu}dx,
\end{align}
and $\mu^*$ is the unique solution of the nonlinear equation 
\begin{align}
\int f_0(x) \log[q(\mu)(x+\mu)]dx=d,
\end{align}
which can be solved numerically (e.g., Newton-Raphson method).

\vspace{-0.4cm}
\section{proof of Theorem \ref{th2}}\label{ap2}

Since the asymmetric KL divergence is also convex with respect to $f(x)$ \cite{MOU}, the optimization problem remains convex; we formulate the Lagrange function i.e. 
\vspace{-0.2cm}
\begin{align}
&L_2=\int x f(x)dx \nonumber \\
&+s \left(\int\! f(x) \log \frac{f(x)}{f_0(x)}dx-d   \right)\!+\!\mu \left(\int\! f(x)dx-1 \!\right).
\end{align}      
\vspace{-0.2cm}
The associated KKT conditions \cite{BOY} are written as       
\begin{align}
x+s\left(\log\frac{f(x)}{f_0(x)}  \right)+\mu=0&, \\
\int f(x)dx=1&, \\
s\left(\int f(x) \log \frac{f(x)}{f_0(x)}dx-d \right)=0&. \label{d1} \\
s\geq 0&.
\end{align}  
For $s>0$, the optimal solution is located at the boundary (complementary slackness in \eqref{d1}); by combining the above equations, we have      
\begin{align}
f(x)=\frac{\psi_0(x,s^*)}{\psi_1(s^*)},
\end{align}
where 
\begin{align}
&\psi_0(x,s)=\exp(-x/s)f_0(x), \\
&\psi_1(s)=\int \psi_0(x,s)dx, \\
&\zeta(s)=\int x \psi_0(x,s)dx.
\end{align}
The optimal dual variable $s^*$ can be found by solving numerically the following equality
\begin{align}
-\frac{\zeta(s)}{\psi_1(s)}-s\log \psi_1(s)=sd. \label{s1}
\end{align}
    
\vspace{-0.8cm}
\section{Proof of Theorem \ref{th3}}\label{ap3}

The Lagrangian function for the problem in (P3) is written as
\vspace{-0.5cm}
\begin{align}
&L=\int x f(x)dx \nonumber 
\end{align}

\begin{align}
&+\frac{s}{2}\left(\int f_0(x)\log \frac{f_0(x)}{f(x)}dx
+ \int f(x)\log \frac{f(x)}{f_0(x)}dx \right) \nonumber \\
&+\mu \left(\int f(x)dx-1 \right).
\end{align}
By taking the derivative of $L$ with respect of $f(x)$, we employ the KKT conditions
\begin{align}
x-\frac{s}{2}\left(\frac{f_0(x)}{f(x)}+\log \frac{f_0(x)}{f(x)}\right)+\mu =0&, \label{x0} \\
\int f(x)dx=1&, \label{x1} \\
\frac{s}{2}\!\left(\!\int\! f_0(x)\log \frac{f_0(x)}{f(x)}dx\!+\!\!\!\int\!\!\! f(x)\log \frac{f(x)}{f_0(x)}dx\!-d \!\right)=0&, \label{x2} \\
s\geq 0&.
\end{align}
\vspace{-0.1cm}
By solving \eqref{x0} with respect to $f(x)$, we derive an expression of the actual distribution 
\begin{align}
f(x)=\frac{f_0(x)}{W_0 \left(\exp \left(\frac{2(x+\mu^*)}{s^*} \right) \right)},
\end{align}
where the optimal dual parameters $s^*$ and $\mu^*$ can be computed numerically by solving a (two-dimensional) nonlinear system of equation that is defined by \eqref{x1} and \eqref{x2}.

\end{document}